\documentclass[letterpaper,aps,prl,floatfix,twocolumn,tightenlines,amsmath,amssymb]{revtex4-1}

\usepackage{graphicx}
\usepackage{amssymb}
\usepackage{color}
\usepackage{psfrag}
\usepackage{ifsym}
\usepackage{epstopdf}
\usepackage{dsfont}
\usepackage{amsmath}
\usepackage{multirow} 
\usepackage{paralist}

\usepackage{amsthm}
\usepackage{enumitem}

\usepackage[normalem]{ulem}

\newcommand{\ket}[1] {| #1 \rangle}

\newcommand{\ketbra}[1]{ | #1 \rangle\!\langle #1 |}

\newcommand{\ie} {\emph{i.e.}}

\newcommand{\moy}[1]{\langle #1 \rangle}

\newcommand{\one}{\leavevmode\hbox{\small1\normalsize\kern-.33em1}}

\newcommand{\ba}{\begin{eqnarray}}
\newcommand{\ea}{\end{eqnarray}}
\newtheorem{thm}{Theorem}

\newtheorem*{lem*}{Lemma}

\begin{document}

\title{Nonlinear Bell inequalities tailored for quantum networks}

\author{Denis Rosset$^{1}$, Cyril Branciard$^{2}$, Tomer Jack Barnea$^{1}$, Gilles P\"utz$^{1}$, Nicolas Brunner$^{3}$, Nicolas Gisin$^{1}$}
\affiliation{$^1$Group of Applied Physics, Universit\'e de Gen\`eve, 1211 Gen\`eve, Switzerland \\ $^{2}$Institut N\'eel, CNRS and Universit\'e Grenoble Alpes, 38042 Grenoble Cedex 9, France
\\ $^{3}$D\'epartement de Physique Th\'eorique, Universit\'e de Gen\`eve, 1211 Gen\`eve, Switzerland}

\date{\today}

\begin{abstract}
In a quantum network, distant observers sharing physical resources emitted by independent sources can establish strong correlations, which defy any classical explanation in terms of local variables. We discuss the characterization of nonlocal correlations in such a situation, when compared to those that can be generated in networks distributing independent local variables.
We present an iterative procedure for constructing Bell inequalities tailored for networks: starting from a given network, and a corresponding Bell inequality, our technique provides new Bell inequalities for a more complex network, involving one additional source and one additional observer. The relevance of our method is illustrated on a variety of networks, for which we demonstrate significant quantum violations.
\end{abstract}

\maketitle

Distant observers performing local measurements on a shared entangled quantum state can observe strong correlations, which have no equivalent in classical physics. This phenomenon, termed quantum nonlocality~\cite{bell64,review}, is at the core of quantum theory and represents a key resource for quantum information processing~\cite{acin,CC}.

This remarkable feature is now relatively well understood in the case of observers sharing entangled states originating from a single common source, for which a solid theoretical framework has been established~\cite{review}, and many classes of Bell inequalities have been derived; see e.g.~\cite{facets}. The situation is however very different in the case of quantum networks, which has been far less explored so far. A quantum network features distant observers, as well as several independent quantum sources distributing entangled states to different subsets of observers (see Fig.~\ref{fig:Nplus1loc}). Crucially, by performing joint measurements, observers can correlate distant (and initially fully independent) quantum systems, hence establishing strong correlations across the entire network. 
Characterizing and detecting the nonlocality of such correlations represents a fundamental challenge, which is also highly relevant to the implementation of quantum networks~\cite{kimble} and quantum repeaters~\cite{sangouard}. Only few exploratory works have discussed nonlocal correlations in the simplest networks, such as the scenario of entanglement swapping~\cite{bilocPRL,bilocPRA} and star-shaped networks~\cite{tavakoli}. Others suggested approaching the problem using the framework of causal inference~\cite{fritz,chaves,chaves1,wood,pusey,chaves2}. The communication cost of simulating quantum correlations in entanglement swapping was also discussed~\cite{branciard}. However, it is fair to say that adequate methods are still currently lacking for discussing nonlocal correlations in networks beyond the simplest possible cases.

In this work, we present a simple and efficient method for detecting and characterizing nonlocal correlations in a wide class of networks. Specifically, we give an iterative procedure for constructing Bell inequalities tailored for networks---that is, inequalities satisfied by any correlations generated in a local model that matches the structure of the network under study, with independent random variables for each independent source. Starting from a given network, and a Bell inequality for it, we then construct inequalities for a more complex network, involving one additional source and one additional observer. We illustrate the relevance of our approach considering a variety of simple networks, and demonstrate significant violations in quantum theory. We believe that the simplicity and versatility of our method makes it adequate for starting a systematic exploration of quantum nonlocality in networks.

\begin{figure}[t!]
 \includegraphics[width=\columnwidth]{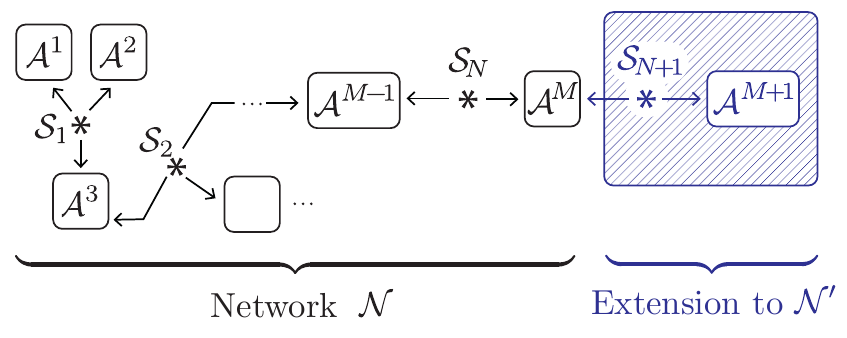}   %{fig_Nplus1loc.pdf}
\caption{
We consider networks consisting of distant observers $\mathcal A^{j}$ sharing physical resources emitted from independent sources $\mathcal S_{k}$, and discuss nonlocality in such networks. In our approach, starting from a network $\mathcal N$ (in black) with $N$ sources and $M$ parties, we define a new network $\mathcal N'$ by adding a new independent source $\mathcal S_{N+1}$ connected to a single party $\mathcal A^{M}$ of $\mathcal N$ and to a new party $\mathcal A^{M+1}$ (in blue). We show how Bell inequalities for so-called $\mathcal N'$-local correlations can be derived starting from Bell inequalities for $\mathcal N$-local correlations.
}
\label{fig:Nplus1loc}
\end{figure}

\paragraph{Scenario of $\mathcal N$-locality.---} % / Definition of $\mathcal N$-locality.---} 

Consider a network $\mathcal N$ consisting of $N$ independent sources $\mathcal S_1, \ldots, \mathcal S_N$ sending physical systems to $M$ parties $\mathcal A^1, \ldots, \mathcal A^M$ (see Fig.~\ref{fig:Nplus1loc}). Each party thus holds a number of systems, and performs a measurement on them (assumed here to be binary). Specifically, we denote by $x^j$ the input received by party $\mathcal A^j$, and by $a^j_{x^j} = \pm 1$ its corresponding output. 

Our goal is to capture the strength of correlations that can be established in a network $\mathcal N$ for different types of resources. In particular, we want to compare the correlations established in the case of a quantum network (\ie, with quantum sources, and with the parties performing quantum measurements), to those that can arise in local (hidden) variable models. Importantly the latter should feature the same network structure as $\mathcal N$, with independent sources of local variables, and are thus referred to as $\mathcal N$-local models. This represents the natural generalization of the notions of Bell locality~\cite{bell64,review} (tailored for the case of a single source), and `bilocality'~\cite{bilocPRL,bilocPRA} (tailored for the scenario of entanglement swapping with two independent sources), to arbitrary networks. 

More formally, we associate to each source $\mathcal S_i$ a random local variable $\lambda_i$, which is sent to all parties connected to $\mathcal S_i$ in the network $\mathcal N$. The crucial assumption of $\mathcal N$-locality is that all $\lambda_i$'s are independent from one another, that is, $\rho(\lambda_1, \ldots \lambda_N) = \prod_i \rho_i(\lambda_i)$, for some (nonnegative and normalised) distributions $\rho_i(\lambda_i)$ over some sets $\Lambda_i$. We denote by $\vec\lambda_{\mathcal A^j}$ the list of random variables $\lambda_i$'s `received' by party $\mathcal A^j$. Then the ($M$-partite) joint probability distribution $P(a^1, \ldots, a^M|x^1, \ldots, x^M)$ (where we have omitted redundant subscripts) is $\mathcal N$-local if and only if it can be decomposed as
\ba
&& P(a^1, \ldots, a^M|x^1, \ldots, x^M) \nonumber \\[1mm]
&& \qquad = \int_{\Lambda_1} \!\!\!\!\! \text{d}\lambda_1 \, \rho_1(\lambda_1) \ldots \!\! \int_{\Lambda_N} \!\!\!\!\!\! \text{d}\lambda_N \, \rho_N(\lambda_N) \nonumber \\[-1mm]
&& \qquad \quad \qquad P(a^1|x^1\!,\vec\lambda_{\mathcal A^1}) \ldots P(a^M|x^M\!,\vec\lambda_{\mathcal A^M}) \, , \quad \label{def:Nloc}
\ea
where each $P(a^j|x^j\!,\vec\lambda_{\mathcal A^j})$ is a valid probability distribution, which (without loss of generality) can be assumed to be deterministic.
As we focus on binary measurements, it is convenient to consider correlators, \ie, the expectation values $\moy{a^1_{x^1} a^2_{x^2} \ldots a^M_{x^M}} $. In a $\mathcal N$-local model, these can be written as
\ba
&& \moy{a^1_{x^1}\ldots a^M_{x^M}} = \int_{\Lambda_1} \!\!\!\!\! \text{d}\lambda_1 \, \rho_1(\lambda_1) \ldots \!\! \int_{\Lambda_N} \!\!\!\!\!\! \text{d}\lambda_N \, \rho_N(\lambda_N) \nonumber \\[-1mm]
&& \hspace{4cm} a^1_{x^1}(\vec\lambda_{\mathcal A^1}) \ldots a^M_{x^M}(\vec\lambda_{\mathcal A^M}), \qquad
\ea
for some deterministic response functions $a^j_{x^j}(\vec\lambda_{\mathcal A^j}) = \pm 1$ of the party's input $x^j$ and of the random variables $\vec\lambda_{\mathcal A^j}$.

Characterizing the set of $\mathcal N$-local correlations is a challenging problem. The main technical difficulty, for cases beyond that of standard Bell locality, originates from the independence of the sources, which makes the set non-convex. Here we will present a simple and efficient technique for generating Bell inequalities tailored for the problem of capturing $\mathcal N$-local correlations. Hence a violation of such inequalities, which is usually possible considering quantum networks, certifies that no $\mathcal N$-local model can reproduce the given correlations. Below we state our main result, which is an iterative procedure for constructing Bell inequalities for $\mathcal N$-local correlations. We then illustrate the relevance of our method by applying it to simple networks, and discuss quantum violations.

\emph{Main result.---} Consider a network $\mathcal N$, and a Bell inequality tailored for it.
From $\mathcal N$, we construct a new network $\mathcal N'$ by adding one source, $\mathcal S_{N+1}$, linked to just one party of $\mathcal N$, say $\mathcal A_M$, and to one new party, $\mathcal A^{M+1}$ (see Fig.~\ref{fig:Nplus1loc}).
The new party $\mathcal A^{M+1}$ gets an input $x^{M+1}$, which we choose to be binary ($x^{M+1} = 0,1$), and gives a binary output $a^{M+1}_{x^{M+1}} = \pm 1$. Given a Bell inequality capturing $\mathcal N$-local correlations, we can now construct a Bell inequality tailored for $\mathcal N'$-local correlations using the following result:

\begin{thm} \label{thm}
Suppose that the correlators $\moy{a^1_{x^1}\ldots a^M_{x^M}}$ in any $\mathcal N$-local model satisfy a Bell inequality of the form
\ba
\sum_{x^1\!, \ldots, x^M} \beta_{x^1\!, \ldots, x^M} \ \moy{a^1_{x^1}\ldots a^M_{x^M}}  \ \le \ 1 \label{ineq:N_loc}
\ea
for some real coefficients $\beta_{x^1\!, \ldots, x^M}$. Then $\mathcal N'$-local correlations (for the network $\mathcal N'$ obtained from $\mathcal N$ as described above) satisfy the following constraint: either there exists $q \in \ ]0,1[$ such that for any partition of the set of party $\mathcal A^M$'s inputs into two disjoint subsets $\mathcal X^M_+$ and $\mathcal X^M_-$, we have
\ba
\frac{1}{q} \, \Sigma_{\mathcal X_+} \, + \, \frac{1}{1-q} \, \Sigma_{\mathcal X_-} \, \le \, 1 \label{ineq:Nplus1_loc}
\ea
for
\ba
\Sigma_{\mathcal X_\pm} =  \sum_{\stackrel{x^1\!, \ldots, x^{M-1}\!,}{x^M \in \mathcal X^M_\pm}} \beta_{x^1\!, \ldots, x^M} \ \moy{a^1_{x^1}\ldots a^M_{x^M} {\textstyle \frac{a^{M+1}_0 \pm a^{M+1}_1}{2}}} \, ; \quad
\ea
or $\Sigma_{\mathcal X_-} = 0$ and $\Sigma_{\mathcal X_+} \le 1$ for all $k$ and all $\mathcal X_\pm^M$; or $\Sigma_{\mathcal X_+} = 0$ and $\Sigma_{\mathcal X_-} \le 1$ for all $k$ and all $\mathcal X_\pm^M$.
\end{thm}

In the present manuscript, we abuse the notation and write $q \in [0,1]$ to cover all cases; indeed, the particular cases where $\Sigma_{\mathcal X_\pm} = 0$ can easily be recovered in the limits $q \to 1$ or $q \to 0$. In Appendix~A we provide a more general statement of the above theorem---which allows one to consider several Bell inequalities at once and also allows for non-full-correlation terms in these inequalities---as well as a detailed proof.
Interestingly, the technique used in our proof also provides an original way to derive the simplest Bell inequality of Clauser-Horne-Shimony-Holt (CHSH)~\cite{chsh}, as discussed in Appendix~B.

\begin{figure*}[hbt!]
 \includegraphics{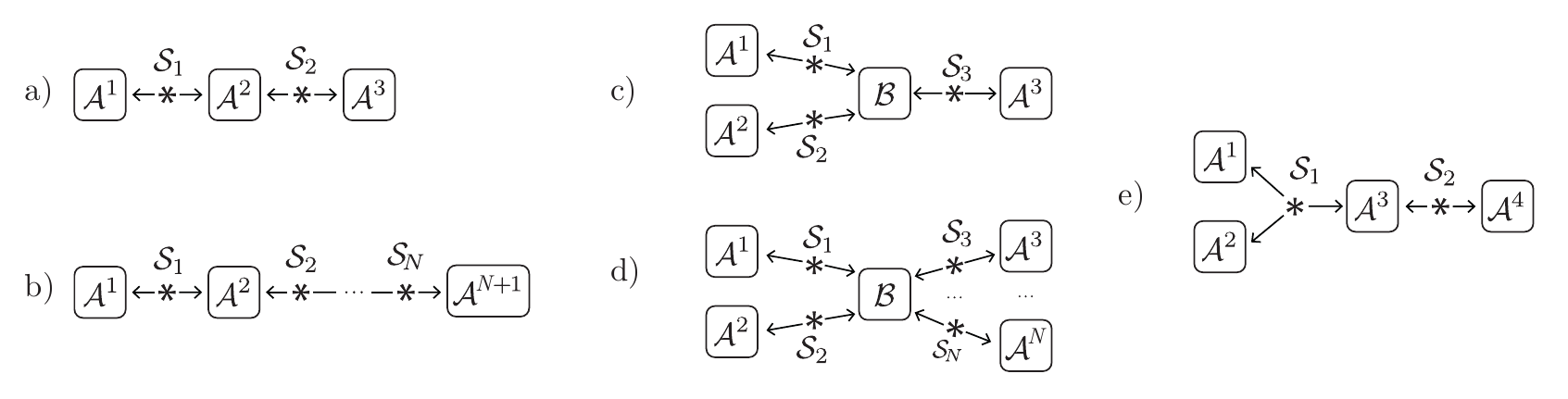}   %{fig_Nplus1loc.pdf}
\caption{In the main text, we discuss a variety of networks: a) the scenario of bilocality, b) a general chain network featuring $N$ sources and $M=N+1$ observers, c) a 3 branch star network, d) a general star-shaped network with $N$ branches, and e) a network featuring a different topology, illustrating the versatility of our method.
}
\label{fig:various_networks}
\end{figure*} 

A remarkable feature of the `Bell inequality'~\eqref{ineq:Nplus1_loc} is that it involves the quantifier `$\exists \, q ...$'. As a consequence, despite its appearance it actually defines a nonlinear constraint on $\mathcal N'$-local correlations. One could eliminate the quantifier by minimizing the left-hand-side of Eq.~\eqref{ineq:Nplus1_loc} over $q$; this would indeed lead to explicitly nonlinear Bell inequalities (see below and Appendices~C--F). However, it will be convenient in general to keep these quantifiers (in a practical test, they could be eliminated later, by optimizing the parameters $q$ directly for the specific values of the observed statistics). In fact, Theorem~\ref{thm} also applies to an initial Bell inequality for $\mathcal N$-local correlations that features quantifiers itself. Our technique can therefore be used in an iterative manner, and allows one to construct Bell inequalities for a broad class of networks, as we shall see below.

\emph{Bilocality.---} Let us first apply the above method to the simplest non-trivial network $\mathcal N$ consisting of $M=2$ parties $\mathcal A^1$ and $\mathcal A^2$ connected to a single source $\mathcal S_1$, that is, the usual Bell scenario~\cite{bell64,review}. In that case, $\mathcal N$-local (\ie~here, simply `Bell local') correlations satisfy the well-known CHSH inequality~\cite{chsh}:
\ba
\moy{{\textstyle \frac{a^1_0+a^1_1}{2}} a^2_0} + \moy{{\textstyle \frac{a^1_0-a^1_1}{2}} a^2_1} \ \le \ 1 \, .
\ea

The network $\mathcal N'$, obtained by adding an independent source $\mathcal S_2$ linked to party $\mathcal A^2$ and to a new party $\mathcal A^3$, corresponds here to the scenario of `bilocality'~\cite{bilocPRL,bilocPRA}; see Fig.~\ref{fig:various_networks}a.
%\begin{center}
% \includegraphics[width=.5\columnwidth]{fig_biloc.pdf}
%\end{center}
Applying Theorem~\ref{thm} starting from the CHSH inequality and with $\mathcal X^2_+ = \{0\}$ and $\mathcal X^2_- = \{1\}$, we find that $\mathcal N'$-local (\ie, bilocal) correlations satisfy the inequality 
\ba
&& \exists \ q \in [0,1] \ \mathrm{such \ that} \nonumber \\%[2mm]
&& \quad {\textstyle \frac{1}{q}} \moy{{\textstyle \frac{a^1_0+a^1_1}{2}} a^2_0 {\textstyle \frac{a_0^3 + a_1^3}{2}}} + {\textstyle \frac{1}{1-q}} \moy{{\textstyle \frac{a^1_0-a^1_1}{2}} a^2_1 {\textstyle \frac{a_0^3 - a_1^3}{2}}} \ \le \ 1 \, . \qquad \label{ineq:biloc}
\ea
It is still fairly easy, in this first example, to eliminate the quantifier. As we show in Appendix~C, this constraint (when combined with similar forms obtained from other versions of CHSH) is equivalent to the (nonlinear) `bilocal inequality' derived previously in~\cite{bilocPRA}. 

Next we discuss the quantum violation of the above Bell inequality, thus considering the entanglement swapping scenario. Assume that each source $\mathcal S_i$ ($i=1,2$) emits two particles in the 2-qubit Werner state $\varrho(v_i) = v_i \ketbra{\Phi^+} + (1{-}v_i) \one/4$, with $v_i \in [0,1]$, $\ket{\Phi^+} = \frac{1}{\sqrt{2}}(\ket{00}+\ket{11})$, and $\one/4$ the fully mixed state of two qubits. Moreover, the parties $\mathcal A^1$ and $\mathcal A^3$ perform single qubit projective measurements given by operators $\hat{a}_0^1 = \hat{a}_0^3 = \frac{\hat\sigma_{\textsc{z}}+\hat\sigma_{\textsc{x}}}{\sqrt{2}}$ (for $x^1,x^3 = 0$) or $\hat{a}_1^1 = \hat{a}_1^3 = \frac{\hat\sigma_{\textsc{z}}-\hat\sigma_{\textsc{x}}}{\sqrt{2}}$ (for $x^1,x^3 = 1$); here $\hat\sigma_{\textsc{z}}$ and $\hat\sigma_{\textsc{x}}$ are the Pauli matrices. Finally, the intermediate party $\mathcal A^2$ performs projective two-qubit measurements given by $\hat{a}_0^2 = \hat\sigma_{\textsc{z}} \otimes \hat\sigma_{\textsc{z}}$ (for $x^2 = 0$) or $\hat{a}_1^2 = \hat\sigma_{\textsc{x}} \otimes \hat\sigma_{\textsc{x}}$ (for $x^2 = 1$). Defining $V = v_1 v_2$, one finds
\ba
\moy{a_{x^1}^1 a_{x^2}^2 a_{x^3}^3} = (-1)^{x^1 x^2 + x^2 x^3} \, \frac{V}{2} \,, \label{eq:triloc:corr}
\ea
so that
\ba
\moy{{\textstyle \frac{a^1_0+a^1_1}{2}} a^2_0 {\textstyle \frac{a^3_0+a^3_1}{2}}} = \moy{{\textstyle \frac{a^1_0-a^1_1}{2}} a^2_1 {\textstyle \frac{a^3_0-a^3_1}{2}}} = \, \frac{V}{2} . \quad
\ea
Noting that $\min_{q \in [0,1]} \big( \frac{1}{q} \frac{V}{2} + \frac{1}{1-q} \frac{V}{2} \big) = 2 \, V$, we find that the quantum correlations thus obtained violate the Bell inequality~\eqref{ineq:biloc}---and hence is non-bilocal---for any $V > \frac12$, as already shown in~\cite{bilocPRL}.

\emph{Chain network.---} The above procedure can be iterated in order to characterize $\mathcal N$-local correlations on a one-dimensional chain network (see Fig.~\ref{fig:various_networks}b). First, starting from the previous bilocality network (with 2 sources and 3 parties), we add a new party $\mathcal A^4$, and a source $S_3$ connected to $\mathcal A^3$ and $\mathcal A^4$. 
Applying Theorem~\ref{thm} to the Bell inequality~\eqref{ineq:biloc}, and choosing $\mathcal X^3_+ = \{0\}$ and $\mathcal X^3_- = \{1\}$, we find that `trilocal' correlations satisfy the inequality 
\ba
&& \exists \ q, r \in [0,1] \ \mathrm{such \ that} \nonumber \\%[2mm]
&& \qquad \frac{1}{8}
	 \Big[ {\textstyle \frac{1}{q} \frac{1}{r}} \moy{(a^1_0{+}a^1_1) a^2_0 a^3_0 (a^4_0{+}a^4_1)} \nonumber \\%[2mm]
&& \qquad \quad + {\textstyle \frac{1}{q} \frac{1}{1-r}} \moy{(a^1_0{+}a^1_1) a^2_0 a^3_1 (a^4_0{-}a^4_1)} \nonumber \\%[2mm]
&& \qquad \quad + {\textstyle \frac{1}{1-q} \frac{1}{r}} \moy{(a^1_0{-}a^1_1) a^2_1 a^3_0 (a^4_0{+}a^4_1)} \nonumber \\%[2mm]
&& \qquad \quad - {\textstyle \frac{1}{1-q} \frac{1}{1-r}} \moy{(a^1_0{-}a^1_1) a^2_1 a^3_1 (a^4_0{-}a^4_1)} \Big] \ \le \ 1 \, . \label{ineq:triloc}
\ea
Note that it is in principle possible to write the above constraint without quantifiers, and end up with a nonlinear Bell inequality (as in the case of bilocality above). We discuss this operation in Appendix~D. However, in this case the nonlinear form appears to be extremely cumbersome and of no practical use. 

%\begin{center}
%\includegraphics[width=.6\columnwidth]{fig_Nloc_line.pdf}
%\end{center}

Next, we extend our analysis to chains of arbitrary lengths, focusing on linear Bell inequalities with quantifiers. 
By further iterating the argument, we obtain the following inequality for chains of $N$ independent sources and $M=N+1$ parties:
\ba
&& \exists \ q^2, \ldots, q^N \in [0,1] \ \mathrm{such \ that} \nonumber \\%[2mm]
&& \quad \frac{1}{2^N} \sum_{x^1, \ldots, x^{N+1}} {\textstyle \frac{1}{q^2_{x^2}} \ldots \frac{1}{q^N_{x^N}}} (-1)^{x^1 x^2 + x^2 x^3 + \ldots + x^N x^{N+1}} \nonumber \\[-4mm]
&& \hspace{5.2cm} \moy{a^1_{x^1} \ldots a^{N+1}_{x^{N+1}}} \ \le \ 1 , \nonumber \\ \label{ineq:Nloc_line}
\ea
with for each $j$, $q^j_0 = q^j$ and $q^j_1 = 1-q^j$.

Let us discuss quantum violations. Consider that each source $\mathcal S_i$ sends two particles in the Werner state $\varrho(v_i)$; party $\mathcal A^1$ measures either $\hat{a}_0^1 = \frac{\hat\sigma_{\textsc{z}}+\hat\sigma_{\textsc{x}}}{\sqrt{2}}$ or $\hat{a}_1^1 = \frac{\hat\sigma_{\textsc{z}}-\hat\sigma_{\textsc{x}}}{\sqrt{2}}$; parties $j$, with $2 \le j \le N$ and $j$ even, measure either $\hat{a}_0^j = \hat\sigma_{\textsc{z}} \otimes \hat\sigma_{\textsc{z}}$ or $\hat{a}_1^j = \hat\sigma_{\textsc{x}} \otimes \hat\sigma_{\textsc{x}}$; parties $j$, with $3 \le j \le N$ and $j$ odd, measure either $\hat{a}_0^j = \frac{\hat\sigma_{\textsc{z}}+\hat\sigma_{\textsc{x}}}{\sqrt{2}} \otimes \frac{\hat\sigma_{\textsc{z}}+\hat\sigma_{\textsc{x}}}{\sqrt{2}}$ or $\hat{a}_1^j = \frac{\hat\sigma_{\textsc{z}}-\hat\sigma_{\textsc{x}}}{\sqrt{2}} \otimes \frac{\hat\sigma_{\textsc{z}}-\hat\sigma_{\textsc{x}}}{\sqrt{2}}$; for $N$ even, party $\mathcal A^{N+1}$ measures either $\hat{a}_0^{N+1} = \hat{a}_0^1 = \frac{\hat\sigma_{\textsc{z}}+\hat\sigma_{\textsc{x}}}{\sqrt{2}}$ or $\hat{a}_1^{N+1} = \hat{a}_1^1 = \frac{\hat\sigma_{\textsc{z}}-\hat\sigma_{\textsc{x}}}{\sqrt{2}}$; for $N$ odd, party $\mathcal A^{N+1}$ measures either $\hat{a}_0^{N+1} = \hat\sigma_{\textsc{z}}$ or $\hat{a}_1^{N+1} = \hat\sigma_{\textsc{x}}$.
Defining $V = \prod_{i=1}^N v_i$, one finds
\ba
\moy{a^1_{x^1} \ldots a^{N+1}_{x^{N+1}}} = (-1)^{x^1 x^2 + x^2 x^3 + \ldots + x^N x^{N+1}} \frac{V}{2^{N/2}} \,. \qquad \label{eq:Nloc:corr}
\ea
The left-hand side of inequality~\eqref{ineq:Nloc_line} is then given by
\ba
&& \frac{2^{N/2} \, V}{4^{N-1}}  \sum_{x^2, \ldots, x^N} {\textstyle \frac{1}{q^2_{x^2}} \ldots \frac{1}{q^N_{x^N}}}   \, .
\ea
Noting that $\min_{q^2, \ldots, q^N} \big( \sum_{x^2, \ldots, x^N} {\textstyle \frac{1}{q^2_{x^2}} \ldots \frac{1}{q^N_{x^N}}} \big) = 4^{N-1}$, we find that the quantum correlations thus obtained violates the Bell inequality~\eqref{ineq:Nloc_line}---and hence are non-$\mathcal N$-local---for $V > 2^{-N/2}$. This proves a conjecture made in~\cite{bilocPRA}. Interestingly, note that although the global correlations become very weak for large $N$ and $V<1$, their nonlocality can nevertheless be revealed using the Bell inequality~\eqref{ineq:Nloc_line}.

\emph{Star network.---} To discuss star-shaped networks, we start from the bilocality network, \ie, a linear chain of 3 parties connected by 2 sources. For clarity, we re-label the parties by calling $\mathcal A^1$ and $\mathcal A^2$ the first and last parties in the chain, and $\mathcal B$ the middle one. The input and output of $\mathcal B$ are now denoted by $y$ and $b_y=\pm1$, respectively. Clearly, $\mathcal N$-local correlations satisfy the Bell inequality~\eqref{ineq:biloc}, with $a^2_{x^2}$ replaced by $b_y$ and $a^3_{x^3}$ replaced by $a^2_{x^2}$.

%\begin{center}
% \includegraphics[width=.4\columnwidth]{fig_triloc_star.pdf}
%\end{center}
%
%\begin{center}
% \includegraphics[width=.3\columnwidth]{fig_Nloc_star.pdf}
%\end{center}

Similarly to our previous constructions, let us add a source $S_3$, connected now to party $\mathcal B$ and to a new party $\mathcal A^3$. The network $\mathcal N'$ thus obtained has a 3-branch star shape (see Fig.~\ref{fig:various_networks}c).
Applying Theorem~\ref{thm} to the Bell inequality of Eq.~\eqref{ineq:biloc} (and with the two subsets of party $\mathcal B$'s inputs $\mathcal Y_+ = \{0\}$ and $\mathcal Y_- = \{1\}$), we find that $\mathcal N'$-local correlations satisfy the inequality
\ba
&& \exists \ q, r \in [0,1] \ \mathrm{such \ that} \nonumber \\%[2mm]
&& \qquad {\textstyle \frac{1}{q} \frac{1}{r}} \, \moy{{\textstyle \frac{a^1_0+a^1_1}{2}} {\textstyle \frac{a^2_0 + a^2_1}{2}} {\textstyle \frac{a^3_0 + a^3_1}{2}} b_0} \nonumber \\%[2mm]
&& \qquad + \ {\textstyle \frac{1}{1-q} \frac{1}{1-r}} \,  \moy{{\textstyle \frac{a^1_0-a^1_1}{2}} {\textstyle \frac{a^2_0 - a^2_1}{2}} {\textstyle \frac{a^3_0 - a^3_1}{2}} b_1} \ \le \ 1 \, . \qquad
\ea

Iterating the above procedure, we obtain a star-shaped network $\mathcal N$ consisting of $N$ independent sources $\mathcal S_i$, each connected to one out of $N$ parties $\mathcal A^i$ and to a single central party $\mathcal B$, as depicted in Fig.~\ref{fig:various_networks}d. For such a network, we find that $\mathcal N$-local correlations satisfy the inequality
\ba
&& \exists \ q_1, \ldots, q_{N-1} \in [0,1] \ \mathrm{such \ that} \nonumber \\%[2mm]
&& \quad {\textstyle \frac{1}{q_1} \ldots \frac{1}{q_{N-1}}} \, \moy{{\textstyle \frac{a^1_0+a^1_1}{2} \ldots \frac{a^N_0 + a^N_1}{2}} b_0} \nonumber \\%[2mm]
&& \quad + \ {\textstyle \frac{1}{1-q_1} \ldots \frac{1}{1-q_{N-1}}} \, \moy{{\textstyle \frac{a^1_0-a^1_1}{2} \ldots \frac{a^N_0 - a^N_1}{2}} b_1} \ \le \ 1 \, . \qquad \label{ineq:star}
\ea
As shown in Appendix~E, by eliminating the quantifiers one can recover here the nonlinear Bell inequalities derived in~\cite{tavakoli}, which generalize the bilocal inequalities of~\cite{bilocPRA} to the star-shaped network considered here. For violations of these inequalities in quantum theory, we refer the reader to Ref.~\cite{tavakoli}.

\emph{Other topologies.---} To illustrate the versatility of our framework, we now discuss a network which is neither a linear chain nor star-shaped. Specifically, we start from a network $\mathcal N$ consisting of a single source $\mathcal S_1$ connected to 3 parties $\mathcal A^1$, $\mathcal A^2$ and $\mathcal A^3$. Here, $\mathcal N$-local (\ie, Bell-local) correlations satisfy the Mermin inequality~\cite{mermin}:
\ba
\moy{{\textstyle \frac{a^1_0 a^2_1+a^1_1 a^2_0}{2}} a^3_0} + \moy{{\textstyle \frac{a^1_0 a^2_0-a^1_1 a^2_1}{2}} a^3_1} \ \le \ 1 \, .
\ea
Adding a source $\mathcal S_2$, linked to party $\mathcal A^3$ and to a new party $\mathcal A^4$, we obtain a network $\mathcal N'$ sketched in Fig.~\ref{fig:various_networks}e. Using Theorem~\ref{thm} and choosing $\mathcal X^3_+ = \{0\}$ and $\mathcal X^3_- = \{1\}$ we find that $\mathcal N'$-local correlations have to obey the following Bell inequality:
\ba
&& \exists \ q \in [0,1] \ \mathrm{such \ that} \nonumber \\%[2mm]
&& \quad {\textstyle \frac{1}{q}} \, \moy{{\textstyle \frac{a^1_0 a^2_1+a^1_1 a^2_0}{2}} a^3_0 {\textstyle \frac{a_0^4 + a_1^4}{2}}} + {\textstyle \frac{1}{1-q}} \, \moy{{\textstyle \frac{a^1_0 a^2_0-a^1_1 a^2_1}{2}} a^3_1 {\textstyle \frac{a_0^4 - a_1^4}{2}}} \ \le \ 1 \, . \nonumber \\ \label{ineq:merminloc}
\ea

%\begin{center}
% \includegraphics[width=.4\columnwidth]{fig_Nloc_Mermin.pdf}
%\end{center}

Let us again discuss quantum violations. Consider that $\mathcal S_1$ sends a noisy 3-qubit state: $\rho_1(v_1) = v_1 \ketbra{GHZ} + (1{-}v_1) \one/8$ with $\ket{GHZ} = \frac{1}{\sqrt{2}}(\ket{000}+\ket{111})$ and where $\one/8$ is the fully mixed state of three qubits, while $\mathcal S_2$ sends 2-qubit Werner state $\varrho_2(v_2)$ as defined previously.
Take for instance the following measurements: for parties $\mathcal A^1$, $\mathcal A^2$ and $\mathcal A^4$, $\hat{a}_0^1 = \hat{a}_0^2 = \hat{a}_0^4 = \frac{\hat\sigma_{\textsc{x}}+\hat\sigma_{\textsc{y}}}{\sqrt{2}}$ and $\hat{a}_1^1 = \hat{a}_1^2 = \hat{a}_1^4 = \frac{\hat\sigma_{\textsc{x}}-\hat\sigma_{\textsc{y}}}{\sqrt{2}}$; for $\mathcal A^3$, $\hat{a}_0^3 = \hat\sigma_{\textsc{x}} \otimes \hat\sigma_{\textsc{x}}$ and $\hat{a}_1^3 = \hat\sigma_{\textsc{y}} \otimes \hat\sigma_{\textsc{y}}$. We then get
\ba
\moy{{\textstyle \frac{a^1_0 a^2_1+a^1_1 a^2_0}{2}} a^3_0 {\textstyle \frac{a_0^4 + a_1^4}{2}}} = \moy{{\textstyle \frac{a^1_0 a^2_0-a^1_1 a^2_1}{2}} a^3_1 {\textstyle \frac{a_0^4 - a_1^4}{2}}} = \, \frac{V}{\sqrt{2}} \, , \quad \nonumber
\ea
where $V = v_1 v_2$. Minimizing again over $q$, we find that quantum correlations violate the Bell inequality~\eqref{ineq:merminloc} when $V > \frac{1}{2\sqrt{2}}$.

\emph{Discussion.---} We presented a simple and efficient method for generating Bell inequalities tailored for networks with independent sources. The relevance of our method was illustrated with various examples, featuring strong quantum violations. 

While we focused here on the case of binary inputs and outputs for each observer, our technique can also be used for deriving Bell inequalities with more inputs and outputs.
In fact, the only requirements that we explicitly made use of is that party $\mathcal A^M$ has binary outputs and the added observer $\mathcal A^{M+1}$ has binary inputs and outputs. In Appendix~F, we illustrate for instance a case with ternary inputs for parties $\mathcal A^1$ and $\mathcal A^2$ in the bilocality scenario, which also includes non-full-correlation terms. In principle our technique could also allow for any numbers of outputs for parties $\mathcal A^1, \ldots \mathcal A^{M-1}$; it would just become quite cumbersome to write without resorting to correlators.
Extending our method to the case where the party $\mathcal A^M$ has more outputs, and party $\mathcal A^{M+1}$ has an arbitrary number of inputs and outputs, is left for future work.

Finally, it would be interesting to derive Bell inequalities tailored for networks featuring loops. In the present work we could only discuss acyclic networks, as our method allows us to `add a leaf' to a graph, but not to create a cycle. Note however that, given a Bell inequality tailored for a network with a loop, our method can readily be applied in order to add a leaf; however, we are not aware of any nontrivial Bell inequality for networks containing a loop, despite intense research efforts in this direction~\cite{bilocPRA,fritz}.

\emph{Note added.---} While writing up this manuscript, we became aware of related work by Lee and Spekkens~\cite{lee} and Chaves~\cite{planet}, discussing polynomial Bell inequalities for networks.

\emph{Acknowledgements.---} We thank Ben Garfinkel for useful comments on the manuscript.
CB acknowledges financial support from the `Retour Post-Doctorants' program (ANR-13-PDOC-0026) of the French National Research Agency and from a Marie Curie International Incoming Fellowship (PIIF-GA-2013-623456) of the European Commission. NB from the Swiss National Science Foundation (grant PP00P2\_138917 and Starting Grant DIAQ), SEFRI (COST action MP1006) and the EU SIQS.

%\newpage

\appendix

\section{Appendix~A: \\ General statement and proof of our main theorem}
\renewcommand{\theequation}{A\arabic{equation}}
\setcounter{equation}{0}

Below we give the full version of our main theorem and its proof. Consider a network $\mathcal N$ with $M$ parties $\mathcal A^j$ and $N$ independent source $\mathcal S_i$, and a Bell inequality tailored for it, with inputs $x^j$ and binary outputs $a^j_{x^j} = \pm 1$. The new network $\mathcal N'$ is obtained by adding one source, $\mathcal S_{N+1}$, and one new party $\mathcal A^{M+1}$. $\mathcal S_{N+1}$ is linked to just one party of $\mathcal N$, say $\mathcal A_M$, and to the new party, $\mathcal A^{M+1}$. 
The inputs and outputs of $\mathcal A^{M+1}$ are both taken to be binary, labeled by $x^{M+1} = 0,1$ and $a^{M+1}_{x^{M+1}} = \pm 1$ respectively. Here in order to also consider non-full-correlation terms, we introduce, for each party, a `trivial' input $x^j = \emptyset$ with a corresponding trivial output $a^j_{\emptyset} = 1$; this will allow us to write for instance $\moy{a^1_{x^1}a^2_{\emptyset}\ldots a^M_{\emptyset}} = \moy{a^1_{x^1}}$.
% (with also $\moy{a^1_{\emptyset}\ldots a^M_{\emptyset}} = 1$)

Given a set of constraints capturing $\mathcal N$-local correlations, we obtain novel constraints for $\mathcal N'$-local correlations as follows.

\begin{thm} \label{thm2}
Suppose that the correlators $\moy{a^1_{x^1}\ldots a^M_{x^M}}$ in any $\mathcal N$-local model satisfy a set of Bell inequalities of the form
\ba
\Big\{ 
\sum_{x^1\!, \ldots, x^M} \beta^{(k)}_{x^1\!, \ldots, x^M} \ \moy{a^1_{x^1}\ldots a^M_{x^M}}  \ \le \ L^{(k)} \Big\}_k \, , \label{ineq:N_loc_ap}
\ea
for some real coefficients $\beta^{(k)}_{x^1\!, \ldots, x^M}$, some `$\mathcal N$-local bounds' $L^{(k)}$, and different values of~$k$.
Then $\mathcal N'$-local correlations (for the network $\mathcal N'$ obtained from $\mathcal N$ as described above) satisfy the following constraint: either there exists $q \in \ ]0,1[$ such that for all $k$ and for any partition of the set of party $\mathcal A^M$'s nontrivial inputs into two disjoint subsets $\mathcal X^M_+$ and $\mathcal X^M_-$,
\ba
\frac{1}{q} \, \Sigma^{(k)}_{\mathcal X_+} \, + \, \frac{1}{1-q} \, \Sigma^{(k)}_{\mathcal X_-} \, + \, \Sigma^{(k)}_{\emptyset} \ \le \ L^{(k)} \label{ineq:Nplus1_loc_ap}
\ea
for
\ba
&&\Sigma^{(k)}_{\mathcal X_\pm} =  \sum_{\stackrel{x^1\!, \ldots, x^{M-1}\!,}{x^M \in \mathcal X^M_\pm}} \beta_{x^1\!, \ldots, x^M}^{(k)} \ \moy{a^1_{x^1}\ldots a^M_{x^M} {\textstyle \frac{a^{M+1}_0 \pm a^{M+1}_1}{2}}}, \nonumber \\[-7mm] \label{def:Sigma_pm_k} \\[3mm]
&& \Sigma^{(k)}_{\emptyset} =  \sum_{x^1\!, \ldots, x^{M-1}} \! \beta^{(k)}_{x^1\!, \ldots, x^{M-1}\!,\emptyset} \, \moy{a^1_{x^1}\ldots a^{M-1}_{x^{M-1}} } \, ; \quad \ \label{def:Sigma_emptyset_k}
\ea
or $\Sigma^{(k)}_{\mathcal X_-} = 0$ and $\Sigma^{(k)}_{\mathcal X_+} + \Sigma^{(k)}_{\emptyset} \le L^{(k)}$ for all $k$ and all $\mathcal X_\pm^M$; or $\Sigma^{(k)}_{\mathcal X_+} = 0$ and $\Sigma^{(k)}_{\mathcal X_-} + \Sigma^{(k)}_{\emptyset} \le L^{(k)}$ for all $k$ and all $\mathcal X_\pm^M$.
\end{thm}

Theorem~\ref{thm2} is a generalization of Theorem~\ref{thm} as given in the main text. It allows for terms that are not `full correlators', and tells us that if one wants to apply Theorem~\ref{thm} to different initial Bell inequalities, the same value of $q$ can be used for all those inequalities.

\begin{proof}[Proof of Theorem~\ref{thm2}]

Consider an $\mathcal N'$-local model with independent random variables $\lambda_i$ ($1 \leq i \leq N$) attached to the $N$ sources $\mathcal S_i$ (as in the general description of a $\mathcal N$-local model in the main text) and an independent random variable $\mu \in \rotatebox[origin=c]{180}{W}$, distributed according to $\rho_{\rotatebox[origin=c]{180}{\tiny W}}(\mu)$, attached to the source $\mathcal S_{N+1}$. (We call it $\mu$ rather than $\lambda_{N+1}$ to ease notation, and to highlight the particular role it plays in our construction.)
We assume, without loss of generality, that the model is deterministic (as any randomness used locally could be included in the variables $\lambda_i$), with binary response functions $a^j_{x^j}(\vec\lambda_{\mathcal A^j}) = \pm 1$, $a^M_{x^M}(\vec\lambda_{\mathcal A^M},\mu) = \pm 1$ and $a^{M+1}_{x^{M+1}}(\mu) = \pm 1$ for parties $\mathcal A^j$ with $1 \leq j \leq M{-}1$, $\mathcal A^M$ and $\mathcal A^{M+1}$, respectively.

Let us define
\ba
& \rotatebox[origin=c]{180}{W}_\pm = \big\{ \mu \in \rotatebox[origin=c]{180}{W} \, \big| \, a^{M+1}_0(\mu) = \pm a^{M+1}_1(\mu) \big\} \, , \quad \\[1mm]
& q_\pm = \int_{\rotatebox[origin=c]{180}{\tiny W}_\pm}  \text{d}\mu \, \rho_{\rotatebox[origin=c]{180}{\tiny W}}(\mu) \, , \quad \rho_{\rotatebox[origin=c]{180}{\tiny W}_\pm}(\mu) = \rho_{\rotatebox[origin=c]{180}{\tiny W}}(\mu) / q_\pm \, ,
\ea
such that $q_+ + q_- = 1$, and so that $\rho_{\rotatebox[origin=c]{180}{\tiny W}_\pm}$ define normalized measures on $\rotatebox[origin=c]{180}{W}_\pm$, resp. (if $q_\pm = 0$, we let $\rho_{\rotatebox[origin=c]{180}{\tiny W}_\pm}(\mu)$ be any normalized measure on $\rotatebox[origin=c]{180}{W}_\pm$).

Let us then calculate, for this $\mathcal N'$-local model:
\ba
&& \moy{a^1_{x^1}\ldots a^M_{x^M} {\textstyle \frac{a^{M+1}_0 \pm a^{M+1}_1}{2}}} \nonumber \\[1mm]
&& = \int_{\Lambda_1} \!\!\!\!\! \text{d}\lambda_1 \, \rho_1(\lambda_1) \ldots \!\! \int_{\Lambda_N} \!\!\!\!\!\! \text{d}\lambda_N \, \rho_N(\lambda_N) \int_{\rotatebox[origin=c]{180}{\tiny W}} \!\! \text{d}\mu \, \rho_{\rotatebox[origin=c]{180}{\tiny W}}(\mu) \nonumber \\[-1mm]
&& \ \ a^1_{x^1}\!(\vec\lambda_{\mathcal A^1}) \ldots a^{M\textrm{-}1}_{x^{M\textrm{-}1}}\!(\vec\lambda_{\mathcal A^{M\textrm{-}1}}) a^M_{x^M}(\vec\lambda_{\mathcal A^M},\mu) {\textstyle \frac{a^{\!M\!+\!1\!}_0(\mu) \pm a^{\!M\!+\!1\!}_1(\mu)}{2}} \nonumber \\[1mm]
&& = q_\pm \int_{\Lambda_1} \!\!\!\!\! \text{d}\lambda_1 \, \rho_1(\lambda_1) \ldots \!\! \int_{\Lambda_N} \!\!\!\!\!\! \text{d}\lambda_N \, \rho_N(\lambda_N) \int_{\rotatebox[origin=c]{180}{\tiny W}_\pm} \!\! \text{d}\mu_\pm \, \rho_{\rotatebox[origin=c]{180}{\tiny W}_\pm}(\mu_\pm) \nonumber \\[-1mm]
&& \qquad a^1_{x^1}(\vec\lambda_{\mathcal A^1}) \ldots a^{M-1}_{x^{M-1}}(\vec\lambda_{\mathcal A^{M-1}}) \tilde a^M_{x^M\!,\mu_\pm}(\vec\lambda_{\mathcal A^M}) \nonumber \\[1mm]
&& = q_\pm \int_{\rotatebox[origin=c]{180}{\tiny W}_\pm} \!\! \text{d}\mu_\pm \, \rho_{\rotatebox[origin=c]{180}{\tiny W}_\pm}(\mu_\pm) \ \moy{a^1_{x^1}\ldots a^{M-1}_{x^{M-1}} \tilde a^M_{x^M\!,\mu_\pm}}, \label{eq:Nplus1_to_N}
\ea
with the (deterministic) response function $\tilde a^M_{x^M\!,\mu_\pm}(\vec\lambda_{\mathcal A^M}) = a^M_{x^M}(\vec\lambda_{\mathcal A^M},\mu_\pm) a^{M+1}_0(\mu_\pm)$, where we relabeled $\mu \to \mu_\pm$ (which is formally understood as an additional input for party $\mathcal A^M$) depending on whether $\mu \in \rotatebox[origin=c]{180}{W}_\pm$, and where the correlator $\moy{a^1_{x^1}\ldots a^{M-1}_{x^{M-1}} \tilde a^M_{x^M\!,\mu_\pm}}$ is $\mathcal N$-local.

Suppose now that the correlators $\moy{a^1_{x^1}\ldots a^M_{x^M}}$ of any $\mathcal N$-local model satisfy the set of Bell inequalities~\eqref{ineq:N_loc_ap} (which may already involve quantifiers of the form `$\exists \, q  \ldots$'). Let us divide the set of party $\mathcal A^M$'s nontrivial inputs into two disjoint subsets $\mathcal X^M_+$ and $\mathcal X^M_-$. By relabeling party $\mathcal A^M$'s inputs as in $x^M \to (x^M\!,\mu_+)$ if $x^M \in \mathcal X^M_+$ and $x^M \to (x^M\!,\mu_-)$ if $x^M \in \mathcal X^M_-$, and by writing now the corresponding outputs as $\tilde a^M_{x^M\!,\mu_\pm}$, it clearly follows that for all $\mu_+$ and $\mu_-$, $\mathcal N$-local correlators $\moy{a^1_{x^1}\ldots a^{M-1}_{x^{M-1}} \tilde a^M_{x^M\!,\mu_\pm}}$ satisfy, for all $k$,
\ba
 S^{(k)}_{\mathcal X_+,\mu_+} +  S^{(k)}_{\mathcal X_-,\mu_-} + \Sigma^{(k)}_{\emptyset} \ \le \ L^{(k)},  \label{ineq:BI_mu}
\ea
where 
\ba
&& S^{(k)}_{\mathcal X_\pm,\mu_\pm} = \!\!\! \sum_{\stackrel{x^1\!, \ldots, x^{M-1}\!,}{x^M \in \mathcal X^M_\pm}} \! \beta^{(k)}_{x^1\!, \ldots, x^M} \, \moy{a^1_{x^1}\ldots a^{M-1}_{x^{M-1}} \tilde a^M_{x^M\!,\mu_\pm}} \, , \qquad \label{ineq:BI_mu_S}
\ea
and with $\Sigma^{(k)}_{\emptyset}$ as defined in Eq.~\eqref{def:Sigma_emptyset_k}.
By Eqs.~\eqref{def:Sigma_pm_k} and~\eqref{eq:Nplus1_to_N}, $S^{(k)}_{\mathcal X_\pm,\mu_\pm}$ are such that
\ba
\Sigma^{(k)}_{\mathcal X_\pm} = q_\pm \int_{\rotatebox[origin=c]{180}{\tiny W}_\pm} \!\! \text{d}\mu_\pm \, \rho_{\rotatebox[origin=c]{180}{\tiny W}_\pm}(\mu_\pm) \, S^{(k)}_{\mathcal X_\pm,\mu_\pm} \, . \label{eq:Sigma_pm_k}
\ea

\medskip

Consider first the case where $q_\pm \in \ ]0, 1[$. Averaging Eq.~\eqref{ineq:BI_mu} over $\mu_+ \in \rotatebox[origin=c]{180}{W}_+$ and $\mu_- \in \rotatebox[origin=c]{180}{W}_-$, and using Eq.~\eqref{eq:Sigma_pm_k} (after divinding it by $q_\pm$), we recover Eq.~\eqref{ineq:Nplus1_loc_ap} with $q=q_+$.

\medskip

In the case where $q_+ = 1$ and $q_- = 0$, first note that the Bell inequalities~\eqref{ineq:N_loc_ap} for $\mathcal N$-local correlations also imply the Bell inequalities
\ba
&& \sum_{\stackrel{x^1\!, \ldots, x^{M-1}\!,}{x^M \in \mathcal X^M_+}} \beta_{x^1\!, \ldots, x^M}^{(k)} \ \moy{a^1_{x^1}\ldots a^M_{x^M}} \nonumber \\
&& \quad - \sum_{\stackrel{x^1\!, \ldots, x^{M-1}\!,}{x^M \in \mathcal X^M_-}} \beta_{x^1\!, \ldots, x^M}^{(k)} \ \moy{a^1_{x^1}\ldots a^M_{x^M}} \nonumber \\
&& \quad + \sum_{x^1\!, \ldots, x^{M-1}} \beta_{x^1\!, \ldots, x^{M-1}\!,\emptyset}^{(k)} \ \moy{a^1_{x^1}\ldots a^{M-1}_{x^{M-1}}} \ \le \ L^{(k)} \, , \qquad \quad
\ea
where we simply changed the sign in front of the coefficients $\beta_{x^1\!, \ldots, x^M}$ with $x^M \in \mathcal X^M_-$ (this can indeed be seen by letting party $\mathcal A^M$ locally flip their outputs when their inputs are in $\mathcal X^M_-$.) Averaging the two families of Bell inequalities, replacing $a^M_{x^M}$ by $\tilde a^M_{x^M\!,\mu_+}$ (as in~\eqref{ineq:BI_mu}--\eqref{ineq:BI_mu_S} above), averaging over $\mu_+ \in \rotatebox[origin=c]{180}{W}_+$ and using Eq.~\eqref{eq:Sigma_pm_k}, we find that $\Sigma^{(k)}_{\mathcal X_-} = 0$ and $\Sigma^{(k)}_{\mathcal X_+} + \Sigma^{(k)}_{\emptyset} \le L^{(k)}$.

Similarly, in the case where $q_+ = 0$ and $q_- = 1$, we find that $\Sigma^{(k)}_{\mathcal X_+} = 0$ and $\Sigma^{(k)}_{\mathcal X_-} + \Sigma^{(k)}_{\emptyset} \le L^{(k)}$, which concludes the proof of Theorem~\ref{thm2}.

\end{proof}

\section{Appendix~B: \\ Obtaining CHSH from the proof of Theorem~\ref{thm2}} 
\renewcommand{\theequation}{B\arabic{equation}}
\setcounter{equation}{0}

Interestingly, our proof of Theorem~\ref{thm2} above provides a way to derive the well-known CHSH Bell inequality~\cite{chsh}, in a similar spirit to our iterative construction of $\mathcal N$-local inequalities.

Consider a trivial network $\mathcal N$ consisting of only one party $\mathcal A^1$.
By adding a source and a new party $\mathcal A^2$ to $\mathcal N$ as previously, we obtain a network $\mathcal N'$ that corresponds to the typical scenario of a Bell experiment. According to Eq.~\eqref{eq:Nplus1_to_N}, the correlators obtained in a $\mathcal N'$-local---\ie, simply `Bell-local'~\cite{bell64,review}---model satisfy
\ba
\moy{a_{x^1}^1 {\textstyle \frac{a_0^2 \pm a_1^2}{2}}} &=& q_\pm \int_{\rotatebox[origin=c]{180}{\tiny W}_\pm} \!\! \text{d}\mu_\pm \, \rho_{\rotatebox[origin=c]{180}{\tiny W}_\pm}(\mu_\pm) \ \moy{\tilde a^1_{x^1\!,\mu_\pm}}.
\ea

The integral above corresponds to the average of the quantities $\moy{\tilde a^1_{x^1\!,\mu_\pm}} \in [-1,1]$, and is therefore itself in the interval $[-1,1]$. The combination $\moy{a_0^1 {\textstyle \frac{a_0^2 + a_1^2}{2}}} + \moy{a_1^1 {\textstyle \frac{a_0^2 - a_1^2}{2}}}$ is thus a convex sum of quantities in $[-1,1]$, with nonnegative weights $q_+$ and $q_- = 1 - q_+$. This implies that
\ba
-1 \ \le \ \moy{a_0^1 {\textstyle \frac{a_0^2 + a_1^2}{2}}} + \moy{a_1^1 {\textstyle \frac{a_0^2 - a_1^2}{2}}} \ \le \ 1,
\ea
which is simply the CHSH Bell inequality.

\section{Appendix~C: \\ Recovering the bilocal inequality of Ref.~\cite{bilocPRA}}
\renewcommand{\theequation}{C\arabic{equation}}
\setcounter{equation}{0}

Let us define, as in~\cite{bilocPRA},
\ba
I = \moy{{\textstyle \frac{a^1_0+a^1_1}{2} a_0^2 \frac{a^3_0 + a^3_1}{2}}} \, , \quad
J = \moy{{\textstyle \frac{a^1_0-a^1_1}{2} a_1^2 \frac{a^3_0 - a^3_1}{2}}} \, . \nonumber 
\ea

Note that in our derivation of~\eqref{ineq:biloc}, we could have applied Theorem~\ref{thm2} to the set of all four equivalent versions of CHSH of the form
\ba
\sigma \moy{{\textstyle \frac{a^1_0+a^1_1}{2}} a^2_0} + \tau \moy{{\textstyle \frac{a^1_0-a^1_1}{2}} a^2_1} \ \le \ 1 \, , \label{ineq:4CHSH}
\ea
for any combination of $\sigma, \tau = \pm 1$. This would have led to four different versions of the inequality in Eq.~\eqref{ineq:biloc} (all for the same $q$) with $\moy{{\textstyle \frac{a^1_0+a^1_1}{2} a_0^2 \frac{a^3_0 + a^3_1}{2}}} (=I)$ replaced by $\sigma I$ and $\moy{{\textstyle \frac{a^1_0-a^1_1}{2} a_0^2 \frac{a^3_0 - a^3_1}{2}}} (=J)$ replaced by $\tau J$. Combining the four Bell inequalities thus obtained, we find that bilocal correlations satisfy
\ba
&& \exists \ q \in [0,1] \ \mathrm{such \ that} \  {\textstyle \frac{1}{q}} |I| + {\textstyle \frac{1}{1-q}} |J| \ \le \ 1 \, . \qquad \label{ineq:biloc_bis}
\ea

One easily finds that the minimal value of ${\textstyle \frac{1}{q}} |I| + {\textstyle \frac{1}{1-q}} |J|$ is $(\!\sqrt{|I|}+\!\sqrt{|J|}\,)^2$, obtained for $q = \frac{\sqrt{|I|}}{\sqrt{|I|}+\!\sqrt{|J|}}$ (or for any $q$ if $I=J=0$). Eq.~\eqref{ineq:biloc_bis} is thus equivalent to
\ba
\sqrt{|I|}\,+\!\sqrt{|J|} \, \le \, 1 \, , \label{ineq:biloc_ter}
\ea
which is indeed the bilocal inequality derived in~\cite{bilocPRA}.

\section{Appendix~D: \\ Bell inequality for trilocality in the nonlinear form}
\renewcommand{\theequation}{D\arabic{equation}}
\setcounter{equation}{0}

In the main text, we derived a Bell inequality for the scenario of `trilocality', that is, a chain network with 3 sources and 4 observers. 
This inequality, given in Eq.~\eqref{ineq:triloc}, can be rewritten without quantifiers. We define here
\ba
I = \moy{\textstyle{\frac{a^1_0+a^1_1}{2} a^2_0 a^3_0 \frac{a^4_0+a^4_1}{2}}}, && \quad
J = \moy{\textstyle{\frac{a^1_0+a^1_1}{2} a^2_0 a^3_1 \frac{a^4_0-a^4_1}{2}}}, \nonumber \\
K = \moy{\textstyle{\frac{a^1_0-a^1_1}{2} a^2_1 a^3_0 \frac{a^4_0+a^4_1}{2}}}, && \quad
L = \moy{\textstyle{\frac{a^1_0-a^1_1}{2} a^2_1 a^3_1 \frac{a^4_0-a^4_1}{2}}}, \nonumber
\ea
and using $q_- = 1 - q$, $r_- = 1 - r$, we multiply Eq.~\eqref{ineq:triloc} by
$2 q q_- r r_-$:
\ba
f=2 q q_-r r_- - q_- r_- I - q_-r J - q r_- K + q r L \ge 0 \, , \nonumber \\ \label{ineq:triloc:poly}
\ea
at the price of losing discriminating power when some of the $I,J,K,L= 0$.

The set of correlations is bounded by $f=0$, when $q,r$ are such that $f$ is maximized; the polynomial equation for this
boundary is thus given by the system $\{ f=0, \frac{\partial f}{\partial q} = 0, \frac{\partial f}{\partial r} = 0 \}$. 
The variables $q, r$ can be removed using an elimination ideal~\cite{elimination}, giving a polynomial $W$ of degree 9 with 286 terms. 
Observing that Eq.~\eqref{ineq:triloc:poly} is symmetric under the group $G$ of order 8 generated by $\{I \leftrightarrow -L \}$ and
$\{I \rightarrow K \rightarrow -L \rightarrow J \rightarrow I\}$, the polynomial $W$ can be decomposed using the
primary invariants $f_1 = I + J + K - L$, $f_2 = J K - I L$, $f_3 = I^2 + J^2 + K^2 + L^2$, $f_4 =-I^2 J K + I J^2 L + I K^2 L - J K L^2$ 
and the secondary invariants $g_1 = 1$, $g_2 = I^3 + J^3 + K^3 - L^3$ of the invariant ring $\mathbb{R}[I,J,K,L]^G$~\cite{invariant}:
\ba
W = W_1 \, g_1 + W_2 \, g_2 \ge 0,
\ea
with factors $W_1$, $W_2$:
\ba 
  \hspace{-1mm} W_1 = &  & f_1^8 + f_2 f_1^7 - 18 f_1^7 - 31 f_2 f_1^6 - 6 f_3 f_1^6 + 20 f_1^6 - 11 f_2^2 f_1^5 \nonumber\\
  &  & + 174 f_2 f_1^5 - 6 f_2 f_3 f_1^5 + 74 f_3 f_1^5 + 2 f_4 f_1^5 - 24 f_1^5  \nonumber\\
  &  & + 183 f_2^2 f_1^4 + 11 f_3^2 f_1^4 - 148 f_2 f_1^4 + 130 f_2 f_3 f_1^4 - 52 f_3 f_1^4  \nonumber\\
  &  & - 30 f_4 f_1^4 + 8 f_1^4 + 40 f_2^3 f_1^3 - 496 f_2^2 f_1^3 + 11 f_2 f_3^2 f_1^3 \nonumber\\
  &  & - 60 f_3^2 f_1^3 + 88 f_2 f_1^3 + 45 f_2^2 f_3 f_1^3 - 494 f_2 f_3 f_1^3 + 72 f_3 f_1^3 \nonumber\\
  &  & - 14 f_2 f_4 f_1^3 - 10 f_3 f_4 f_1^3 - 180 f_4 f_1^3 - 312 f_2^3 f_1^2  - 6 f_3^3 f_1^2 \nonumber\\
  &  & + 288 f_2^2 f_1^2 - 117 f_2 f_3^2 f_1^2 - 24 f_3^2 f_1^2 - 24 f_2 f_1^2 \nonumber\\
  &  & - 510 f_2^2 f_3 f_1^2  + 300 f_2 f_3 f_1^2 - 24 f_3 f_1^2 - 108 f_2 f_4 f_1^2 \nonumber\\
  &  & + 90 f_3 f_4 f_1^2 + 120 f_4 f_1^2 - 48 f_2^4 f_1 + 384 f_2^3 f_1 - 6 f_2 f_3^3 f_1  \nonumber\\
  &  & - 42 f_2^2 f_3^2 f_1 + 120 f_2 f_3^2 f_1 - 84 f_2^3 f_3 f_1 - 144 f_2^2 f_1  \nonumber \\
  &  & + 828 f_2^2 f_3 f_1 - 120 f_2 f_3 f_1 + 24 f_2^2 f_4 f_1 + 12 f_3^2 f_4 f_1 \nonumber\\
  &  & + 888 f_2 f_4 f_1 + 36 f_2 f_3 f_4 f_1 + 336 f_3 f_4 f_1 - 144 f_4 f_1  \nonumber\\
  &  & + 48 f_2^4 - 288 f_2^3 + 6 f_2 f_3^3 + 48 f_2^2 + 135 f_2^2 f_3^2 + 24 f_2 f_3^2  \nonumber\\
  &  & - 324 f_4^2 + 432 f_2^3 f_3 - 120 f_2^2 f_3 + 24 f_2 f_3 + 768 f_2^2 f_4 \nonumber\\
  &  &  + 12 f_3^2 f_4 - 336 f_2 f_4 + 84 f_2 f_3 f_4 + 48 f_3 f_4 + 48 f_4 \nonumber
\ea
\ba
  \hspace{-1mm} W_2 & = & 2 f_1^5+2 f_2 f_1^4-36 f_1^4-56 f_2 f_1^3-6 f_3 f_1^3+40 f_1^3  \nonumber \\
&& -16 f_2^2 f_1^2+240 f_2 f_1^2-6 f_2 f_3 f_1^2+40 f_3 f_1^2+4 f_4 f_1^2 \nonumber \\
&& -48 f_1^2+192 f_2^2 f_1+4 f_3^2 f_1-176 f_2 f_1+68 f_2 f_3 f_1 \nonumber \\
&& +16 f_3 f_1-72 f_4 f_1 +16 f_1+32 f_2^3-320 f_2^2+4 f_2 f_3^2 \nonumber \\
&& +32 f_2+24 f_2^2 f_3-40 f_2 f_3-16 f_2 f_4-8 f_3 f_4-144 f_4. \nonumber
\ea

Such a nonlinear form is however clearly too cumbersome for any practical use. Nevertheless, the quantum correlations of Eq.~\eqref{eq:Nloc:corr}, for which $I=J=K=-L=V/2^{3/2}$,
%for which $I=J=K=-L=V/2\sqrt{2}$, give $f_1 = \sqrt{2} \, V$, $f_2 = V^2/4$, $f_3 = V^2/2$, $f_4 =-V^4/16$, $g_1 = 1$, $g_2 = V^3/4\sqrt{2}$, $W_1 = V^4 (-1 + 5 \sqrt{2} \, V - 11 V^2)$ and $W_2 = 4 \sqrt{2} \, V (4 - 11 \sqrt{2} \, V + 11 V^2)$, leading to the inequality $W = 3 V^4 (1 - 2 \sqrt{2} \, V) \ge 0$,
lead to the inequality $W = 3 V^4 (1 - 2^{3/2} \, V) \ge 0$, which is violated for $V > 2^{-3/2}$.

\section{Appendix~E: Recovering the Bell inequalities \\ of Ref.~\cite{tavakoli} for star networks}
\renewcommand{\theequation}{E\arabic{equation}}
\setcounter{equation}{0}

Let us define here, as in~\cite{tavakoli},
\ba
I &\!=& \moy{{\textstyle \frac{a^1_0+a^1_1}{2} \ldots \frac{a^N_0 + a^N_1}{2}} b_0} = \frac{1}{2^N} \! \sum_{x^1,\ldots,x^N} \! \moy{a^1_{x^1} \ldots a^N_{x^N} b_0} , \nonumber \\%[2mm]
J &\!=& \! \moy{{\textstyle \frac{a^1_0 \textrm{-} a^1_1}{2} \!\ldots\! \frac{a^N_0 \!\textrm{-} a^N_1\!}{2}} b_1} \!=\! \frac{1}{2^N} \!\!\! \sum_{x^1\!,\ldots,x^N} \!\!\! (\textrm{-}1)^{\sum \! x^j} \! \moy{a^1_{x^1} \!\ldots\! a^N_{x^N} b_1} . \nonumber 
\ea

By initially starting from all four versions of CHSH of the form~\eqref{ineq:4CHSH}, one can actually obtain a stronger constraint than Eq.~\eqref{ineq:star}, namely
\ba
&& \exists \ q_1, \ldots, q_{N-1} \in [0,1] \ \mathrm{such \ that} \nonumber \\%[2mm]
&& \qquad {\textstyle \frac{1}{q_1} \ldots \frac{1}{q_{N-1}}} \, | I | + {\textstyle \frac{1}{1-q_1} \ldots \frac{1}{1-q_{N-1}}} \, | J | \ \le \ 1 \, . \qquad \label{ineq:N_loc_bis}
\ea
One easily finds that the minimal value of the left hand side of the inequality above, for $q_1, \ldots, q_{N-1} \in [0,1]$, is $(\!\sqrt[N]{|I|}+\!\sqrt[N]{|J|}\,)^N$, obtained for all $q_j = \frac{\sqrt[N]{|I|}}{\sqrt[N]{|I|}+\!\sqrt[N]{|J|}}$ (or for any $q_j$ if $I=J=0$). Eq.~\eqref{ineq:N_loc_bis} is thus equivalent to
\ba
\sqrt[N]{|I|}\,+\!\sqrt[N]{|J|} \, \le \, 1 \, , \label{ineq:N_loc_ter}
\ea
which is indeed the Bell inequality for $\mathcal N$-local correlations in a star-shaped network derived in~\cite{tavakoli}.

\section{Appendix~F: \\ Extending the Bell inequality $I_{3322}$ to bilocality}
\renewcommand{\theequation}{F\arabic{equation}}
\setcounter{equation}{0}

In this appendix, we illustrate how our extension technique can be applied to Bell inequalities which feature more than two inputs as well as marginal terms. To give a specific example, we start from the simple Bell inequality $I_{3322}$~\cite{collinsgisin} in the standard Bell scenario (with two parties $\mathcal{A}^1$ and $\mathcal{A}^2$ sharing a source $\mathcal{S}_1$), expressed here in correlation form:
\ba   
 \moy{a^1_{0} a^2_0} +  \moy{a^1_{0} a^2_1} + \moy{a^1_{0} a^2_2} + \moy{a^1_{1} a^2_0} & &  \nonumber  \\
+ \moy{a^1_{1} a^2_1} -\moy{a^1_{1} a^2_2} +  \moy{a^1_{2} a^2_0} - \moy{a^1_{2} a^2_1}  & &\nonumber  \\
-\moy{a^1_{0}}-\moy{a^1_{1}} +\moy{a^2_{0}}+ \moy{a^2_{1}} &  \leq & 4 \, . \label{eq:I3322}
\ea

Adding a new observer $\mathcal{A}^3$, and a new independent source $\mathcal{S}_2$ connected to $\mathcal{A}^2$ and $\mathcal{A}^3$, we arrive at the scenario of bilocality. Applying Theorem~\ref{thm2} to the above Bell inequality, and choosing for instance $\mathcal X^2_+ = \{0\}$ and $\mathcal X^2_- = \{1,2\}$, we find that `bilocal' correlations satisfy the Bell inequality 
\ba
&& \exists \ q \in [0,1] \ \mathrm{  such \ that} \nonumber \\%[2mm]
&& \quad {\textstyle \frac{1}{q}} \moy{(a^1_{0} {+} a^1_{1} {+} a^1_{2} {+} 1 ) a^2_0 {\textstyle \frac{a_0^3 + a_1^3}{2}}} \nonumber \\
&& \ + {\textstyle \frac{1}{1-q} } \big[ \moy{(a^1_{0} {+} a^1_{1} {-} a^1_{2} {+} 1) a^2_1 {\textstyle \frac{a_0^3 - a_1^3}{2}}} + \moy{(a^1_{0} {-} a^1_{1}) a^2_2 {\textstyle \frac{a_0^3 - a_1^3}{2}}}  \big]   \nonumber \\ && \qquad \qquad \qquad - \moy{a^1_{0}} - \moy{a^1_{1}} \ \le \  4  \, . \qquad
\ea
Following a similar procedure to that discussed in Appendix~C, the above inequality (after adding absolute values as above) can be rewritten without quantifier. We obtain the nonlinear form
\ba
\sqrt{|I'|}\,+\!\sqrt{|J'|} \, \le \, \sqrt{|L'|} \, , 
\ea
where we have defined
\ba 
&& I' = \moy{(a^1_{0} {+} a^1_{1} {+} a^1_{2} {+} 1 ) a^2_0 {\textstyle \frac{a_0^3 + a_1^3}{2}}}\,, \\
&& J' = \moy{(a^1_{0} {+} a^1_{1} {-} a^1_{2} {+} 1) a^2_1 {\textstyle \frac{a_0^3 - a_1^3}{2}}}  + \moy{(a^1_{0} {-} a^1_{1}) a^2_2 {\textstyle \frac{a_0^3 - a_1^3}{2}}}\,, \qquad \\
&& L' =  4 +\moy{a^1_{0}}+\moy{a^1_{1}} \,.
\ea

Note that we could also adopt a different choice for $\mathcal X^2_\pm$, or we could exchange the parties $\mathcal{A}^1$ and $\mathcal{A}^2$ when writing the $I_{3322}$ inequality~\eqref{eq:I3322}, which would result in different Bell inequalities for bilocal correlations.

\end{document}